\documentclass[%
 reprint,
superscriptaddress,
showpacs,preprintnumbers,
 amsmath,amssymb,
 aps,
 prl,
]{revtex4-1}

\usepackage{graphicx}% Include figure files
\usepackage{dcolumn}% Align table columns on decimal point
\usepackage{bm}% bold math
\usepackage{amsmath,amsfonts,amssymb,amsthm}
\usepackage{txfonts}
\usepackage{enumerate}
\usepackage{color}

\newtheorem{thm}{Theorem}
\newtheorem{prop}[thm]{Proposition}

\begin{document}

\preprint{APS/123-QED}

\title{Self-consistent multiple complex-kink solutions in \\ Bogoliubov--de~Gennes and chiral Gross--Neveu systems}
\author{Daisuke A. Takahashi}\email{takahashi@vortex.c.u-tokyo.ac.jp}
\affiliation{Department of Basic Science, The University of Tokyo, Tokyo 153-8902, Japan}
\affiliation{Research and Education Center for Natural Sciences, Keio University, Hiyoshi 4-1-1, Yokohama, Kanagawa 223-8521, Japan}
\author{Muneto Nitta}
\affiliation{Research and Education Center for Natural Sciences, Keio University, Hiyoshi 4-1-1, Yokohama, Kanagawa 223-8521, Japan}
\affiliation{Department of Physics, Keio University, Hiyoshi 4-1-1, Yokohama, Kanagawa 223-8521, Japan}

\date{\today}% It is always \today, today,

\begin{abstract}
	We exhaust all exact self-consistent solutions of complex-valued fermionic condensates 
in the 1+1 dimensional Bogoliubov--de Gennes and chiral Gross--Neveu systems 
under uniform boundary conditions. 
We obtain $n$ complex (twisted) kinks, or grey solitons, 
with $2n$ parameters corresponding to their positions and phase shifts. 
Each soliton can be placed at an arbitrary position while the self-consistency 
requires its phase shift to be quantized by $\pi/N$ for $N$ flavors.
\end{abstract}

\maketitle

%\section{Introduction}
\textit{Introduction.---}%
The search for inhomogeneous self-consistent fermionic condensates 
including states such as the Fulde--Ferrell (FF) \cite{Fulde:1964zz} and 
Larkin--Ovchinnikov (LO) \cite{larkin:1964zz} states 
having phase and amplitude modulations, respectively, in superconductors 
has attracted considerable attentions for more than half a century 
since theoretical predictions were made about their existence. 
While amplitude modulations are well studied in conducting polymers 
\cite{Brazovskii1,Horovitz:1981,Mertsching:1981,Brazovskii2,Brazovskii3},  
the FFLO state is theoretically shown to be a ground state of superconductors 
under a magnetic field \cite{Machida:1984zz}.
Recently, the FFLO state has also been discussed in the context of 
an ultracold atomic Fermi gas \cite{Radzihovsky:2010,Radzihovsky:2011}. 
In general, inhomogeneous self-consistent fermionic condensates 
with a gap function and fermionic excitations 
can be treated simultaneously using 
the Bogoliubov--de Gennes (BdG) and gap equations \cite{DeGennes:1999}. 
The gap functions are real and complex 
for conducting polymers \cite{Heeger:1988zz} and superconductors, 
respectively.
In the quantum field theory, these systems correspond to 
the Gross--Neveu (GN) model \cite{Gross:1974jv} and 
the Nambu--Jona-Lasinio (or chiral GN) model \cite{Nambu:1961tp}, 
which were proposed as models of dynamical chiral symmetry breaking 
in 1+1 or 2+1 dimensions. 
Therefore,  
BdG and (chiral) GN systems 
have been studied and developed together from the viewpoint of both condensed matter physics
 and high energy physics (see Ref.~\cite{Thies:2006ti} for a review). 
For instance, fermion number fractionization 
is one of the topics that has been studied from this viewpoint \cite{Jackiw:1975fn,Niemi:1984vz}.
Recently, it has been shown that 
the solutions in 1+1 dimensions can be  
promoted to 3+1 dimensions \cite{Nickel:2009wj,Hofmann:2010gc}, 
thereby leading to extensive study of the modulated phases of these systems in terms of quantum chromodynamics (QCD) \cite{Casalbuoni:2003wh}.

\indent Inhomogeneous self-consistent solutions are often studied numerically 
because analytic solutions are generally difficult to obtain. 
However, several analytic solutions are available in the case of the real-valued condensates in 1+1 dimensions, 
which describe the conducting polymers and the real GN model. 
Under uniform boundary conditions at spatial infinities, 
a real kink was constructed by Dashen {\it et.~al.} \cite{Dashen:1975xh} 
by using the inverse scattering method, and later, it was 
reconstructed in polyacetylene \cite{Takayama:1980zz} 
in the continuum limit of the lattice model \cite{Su:1979ua}.
Subsequently, a bound state of a kink and an anti-kink, 
which is called a polaron, 
was constructed in polyacetylene \cite{Campbell:1981,Campbell:1981dc}, 
for which achieving self-consistency in the system requires the distance between 
the kink and anti-kink to be fixed. Furthermore, 
three kinks (kink and polaron placed at arbitrary positions) \cite{OkunoOnodera,Feinberg:2002nq} 
and more general solutions \cite{Feinberg:2003qz} were obtained.
The attractive interaction between two polarons was also investigated \cite{OkunoOnodera2}.
For a periodic boundary condition, 
the existence of real kink crystals (the LO state) has been known for a long time 
\cite{Brazovskii1,Horovitz:1981,Mertsching:1981,Brazovskii2,Brazovskii3}.

\indent On the other hand, 
when compared with real condensates, 
only a few self-consistent solutions have thus far been obtained
for complex condensates, 
 such as a complex (or twisted) kink or a grey soliton, \cite{Shei:1976mn} 
and their crystals  \cite{Basar:2008im,Basar:2008ki}. 
In these complex-valued crystals, both the amplitude and phase are modulated (the FFLO state), 
and this modulated phase has important applications in both 
superconductors and QCD, such as 
in the phase diagram of the chiral GN model \cite{Basar:2009fg}.
An attempt to construct more general solutions was made 
\cite{Correa:2009xa,Takahashi:2012aw} 
by using a technique of integrable systems known as the nonlinear Schr\"{o}dinger or 
Ablowitz--Kaup--Newell--Segur hierarchy \cite{AKNS1974}.

\indent In this Letter, we exhaust all exact self-consistent solutions 
of complex condensates
under uniform boundary conditions, 
and we find that they describe multiple twisted kinks. 
Unlike polarons in real condensates, 
where the distance between the kink and anti-kink is fixed,
the situation is drastically simplified in our multiple twisted-kink solutions; 
we determine the filling rate of fermions for bound states of each kink, 
and we find that each kink can be placed at any position 
and has any phase shift quantized by $\pi/N$ with the number of flavors $N$.

%\section{Fundamental equations}
\textit{Fundamental equations.---}%
The fundamental equations which we consider in this Letter appear in both condensed matter and high energy physics. In the condensed matter language, they are the one-dimensional BdG system with the Andreev approximation consisting of the BdG equation for right movers (BdG${}_{\text{R}}$)
	\begin{align}
		\begin{pmatrix} -\mathrm{i}\partial_x & \Delta(x) \\ \Delta(x)^* & \mathrm{i}\partial_x \end{pmatrix}\begin{pmatrix} u_{\text{R}} \\ v_{\text{R}} \end{pmatrix} = \epsilon\begin{pmatrix} u_{\text{R}} \\ v_{\text{R}} \end{pmatrix},
\label{eq:BdG-R}
	\end{align}
	the BdG equation for left movers (BdG${}_{\text{L}}$)
	\begin{align}
		\begin{pmatrix} \mathrm{i}\partial_x & \Delta(x) \\ \Delta(x)^* & -\mathrm{i}\partial_x \end{pmatrix}\begin{pmatrix} u_{\text{L}} \\ v_{\text{L}} \end{pmatrix} = \epsilon\begin{pmatrix} u_{\text{L}} \\ v_{\text{L}} \end{pmatrix},
\label{eq:BdG-L}
	\end{align}
	and the gap equation as a self-consistent condition
	\begin{align}
		-\frac{\Delta(x)}{g} = \sum_{\text{occupied states}}\left( u_{\text{R}}v_{\text{R}}^*+u_{\text{L}}v_{\text{L}}^* \right).
\label{eq:gap}
	\end{align}
For a derivation from the second quantized Hamiltonian, see, e.g., Ref.~\cite{Machida:1984zz}.

In high energy physics, this problem is equivalent to 
the chiral GN model with $N$ flavors, 
\begin{eqnarray}
{\cal L} = \bar \psi \mathrm{i} {\slash \hspace{-1.05ex}\partial} \psi 
+ \frac{g^2}{2N} \left[(\bar \psi \psi)^2 
+ (\bar \psi \mathrm{i} \gamma_5 \psi)^2\right] \label{eq:GN}
\end{eqnarray}
with $\psi(x) =(\psi_1(x),\cdots,\psi_N(x))^T$ \cite{Gross:1974jv,Nambu:1961tp,Dashen:1975xh,Shei:1976mn}.
Introducing the auxiliary fields $\sigma(x)$ and $\pi(x)$, 
this can be rewritten as
\begin{eqnarray}
{\cal L} = \bar \psi \mathrm{i} {\slash \hspace{-1.05ex}\partial} \psi 
- g \bar \psi (\sigma + \mathrm{i} \pi \gamma_5) \psi 
- \frac{N}{2} (\sigma^2 + \pi^2).
\end{eqnarray}
Eliminating $\sigma(x)$ and $\pi(x)$ by their equations of motion, $\sigma = -(g/N) \bar \psi \psi$ and $\pi = -(g/N) \bar \psi \mathrm{i} \gamma_5\psi$, 
takes us back to (\ref{eq:GN}). Instead, we integrate out $\psi(x)$ 
to obtain $ Z=\int {\cal D}\sigma {\cal D}\pi\exp(\mathrm{i}S_{\text{eff}}) $ with
\begin{eqnarray}
 S_{\text{eff}} =  N \left( -\mathrm{i} \ln {\rm Det} 
\left[\mathrm{i} {\slash \hspace{-1.05ex}\partial} 
- g (\sigma + \mathrm{i} \pi \gamma_5) \right]
- \frac{1}{2} (\sigma^2 + \pi^2)
\right).
\end{eqnarray}
Defining $\Delta(x) = \sigma (x) + \mathrm{i} \pi (x)$, 
the gap equation is obtained in the large-$N$ limit 
as the stationary condition for $\Delta^*(x)$
\begin{eqnarray}
 \Delta(x) = - 4 \mathrm{i}
\frac{\delta}{\delta \Delta^*(x)}  
\ln {\rm Det} \left[\mathrm{i} {\slash \hspace{-1.05ex}\partial} 
-g(\sigma + \mathrm{i} \pi \gamma_5 ) \right]. \label{eq:gap2}
\end{eqnarray}
In the Hartree-Fock formalism, we consider 
$H_{\text{R}} \psi_{\text{R}} = \epsilon \psi_{\text{R}}$ and 
$H_{\text{L}} \psi_{\text{L}} = \epsilon \psi_{\text{L}}$
with single-particle Hamiltonians  
$H_{\text{R}} = - \mathrm{i} \gamma_5 \partial_x
+ \gamma_0 (\sigma + \mathrm{i} \pi \gamma_5 )$ and 
$H_{\text{L}} = + \mathrm{i} \gamma_5 \partial_x
+ \gamma_0 (\sigma + \mathrm{i} \pi \gamma_5 )$,  
reducing to 
the BdG Eqs.~(\ref{eq:BdG-R}) and (\ref{eq:BdG-L}) 
with $\gamma_0 = \sigma_1$, $\gamma_1 = -\mathrm{i}\sigma_2$ 
and $\gamma_5 =\sigma_3$,   
while the consistency condition 
$\Delta  = -(g/ N) \left(\left<\bar \psi \psi\right> + \mathrm{i} \left<\bar\psi \mathrm{i} \gamma_5 \psi\right>\right)$ reduces to Eq.~(\ref{eq:gap}).

%\section{Result from the inverse scattering theory}
\textit{Result from the inverse scattering theory.---}%
	First, we briefly summarize the mathematical expressions of the  $ n $-soliton solution and its eigenstates of the self-defocusing Zakharov--Shabat eigenvalue problem \cite{ZakharovShabat2}
	\begin{align}
		\begin{pmatrix} -\mathrm{i}\partial_x & \Delta(x) \\ \Delta(x)^* & \mathrm{i}\partial_x \end{pmatrix}\begin{pmatrix} u \\ v \end{pmatrix} = \epsilon\begin{pmatrix} u \\ v \end{pmatrix} \label{eq:ZS}
	\end{align}
	obtained by the inverse scattering method \cite{FaddeevTakhtajan}. The detailed derivation is provided in the Supplemental Material \footnote{See the Supplemental Material below.}.\\
	\indent Let us assume that the gap function obeys the following asymptotically uniform boundary condition:
	\begin{align}
		|\Delta(x)| \rightarrow m \ (>0), \qquad x\rightarrow\pm\infty. \label{eq:gapNVB}
	\end{align}
	Subsequently, we parametrize the energy and wavenumber of quasiparticles using the uniformizing variable $ s $ defined by
	\begin{align}
		\epsilon(s) = \frac{m}{2}(s+s^{-1}),\quad k(s) = \frac{m}{2}(s-s^{-1}). \label{eq:uniformizer}
	\end{align}
	We can easily verify that the dispersion relation $ \epsilon^2=k^2+m^2 $ holds for an arbitrary complex number $ s $. Eigenstates corresponding to $ s $ on the real axis are scattering states, while those on the unit circle are bound states. Since $ s $ and $ s^* $ on the unit circle correspond to the same bound state, it is sufficient to consider the unit circle in the upper half-plane when we count the number of bound states (see Fig.~\ref{fig:uniformizer}).\\ 
%The ranges of the variable are as follows 
%	\begin{align}
%	\begin{split}
%		-\infty<s<0 &\quad\leftrightarrow\quad \epsilon<0 \ \&  -\infty<k<\infty, \\
%		0<s<\infty &\quad\leftrightarrow\quad \epsilon>0 \ \&   -\infty<k<\infty,
%	\end{split}\label{uniformizer}
%	\end{align}
%	and therefore when we rewrite the integral of $ k $ using $ s $, the interval becomes  $ (-\infty,0) $ for $ \epsilon<0 $ and $ (0,+\infty) $ for $ \epsilon>0 $.\\ 
	\begin{figure}[tb]
		\begin{center}
		\includegraphics{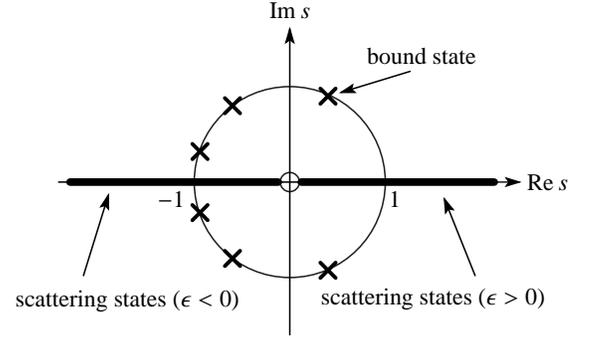}
		\caption{\label{fig:uniformizer} Scattering and bound states in $ s $-plane. Scattering states with positive (negative) energy exist on the real and positive (negative) axis. Bound states exist on the unit circle, and $ s $ and $ s^* $ represent the same bound state.}
		\end{center}
	\end{figure}
	\begin{figure}[tb]
		\begin{center}
		\includegraphics[scale=1.2]{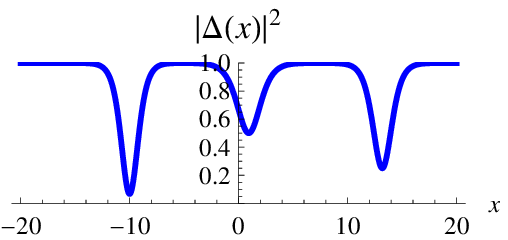} \\ \includegraphics[scale=1.2]{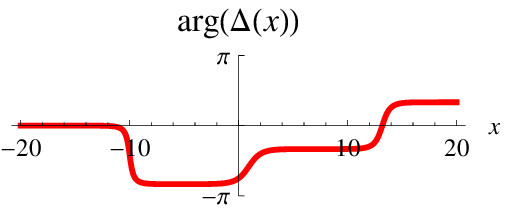}
		\caption{\label{fig:threekink} Example of a three-kink solution. Here the parameters are $ m=1,\,s_1=\mathrm{e}^{\frac{5}{12}\pi\mathrm{i}},\,s_2=\mathrm{e}^{\frac{2}{3}\pi\mathrm{i}},\,s_3=\mathrm{e}^{\frac{3}{4}\pi\mathrm{i}},\,x_1=-10,\,x_2=10, $ and $ x_3=0 $. The positions of the solitons (Eq.~(\ref{eq:SolitonPosition010})) are calculated as $ X_1=-10,\,X_2=13.18 $, and $ X_3=0.93 $. }
		\end{center}
	\end{figure}
	\indent Let us consider the gap function $ \Delta(x) $ which has $ n $ bound states and acts as a reflectionless potential for scattering states, {\it i.e.},  the $ n $-soliton solution. By writing the $ s $ values of the bound states as $ s_j=\mathrm{e}^{\mathrm{i}\theta_j}\; (j=1,\dots,n) $ with $0<\theta_j<\pi$, the eigenenergy and the complex wavenumber can be rewritten as
	\begin{align}
		\kappa_j:=-\mathrm{i}k(s_j)=m\sin\theta_j,\qquad \epsilon_j:=\epsilon(s_j)=m\cos\theta_j,
	\end{align}
respectively. According to the inverse scattering theory, $ \theta_1,\dots,\theta_n $ are all different from each other and there is no degeneracy. We further introduce the following notation:
	\begin{align}
		e_j(x)=\sqrt{\kappa_j}\,\mathrm{e}^{\kappa_j(x-x_j)}, \quad (j=1,\dots,n).
	\end{align}
	Here, the real constant $ x_j $ represents the position of the $ j $-th soliton up to an additive constant when solitons are well separated from each other, as shown below. Furthermore, we define the functions $ f_1(x),\dots,f_n(x) $ as solutions of the following linear equation:
	\begin{align}
		\begin{pmatrix}f_1 \\ f_2 \\ \vdots \\ f_n \end{pmatrix}+\begin{pmatrix}e_1 \\ e_2 \\ \vdots \\ e_n \end{pmatrix}-\frac{2\mathrm{i}}{m}\begin{pmatrix} \frac{e_1^2}{s_1^{-1}-s_1} & \frac{e_1e_2}{s_1^{-1}-s_2} & \dots & \frac{e_1e_n}{s_1^{-1}-s_n} \\ \frac{e_2e_1}{s_2^{-1}-s_1} & \frac{e_2^2}{s_2^{-1}-s_2} & \dots & \frac{e_2e_n}{s_2^{-1}-s_n} \\ \vdots && \ddots & \vdots \\ \frac{e_ne_1}{s_n^{-1}-s_1} & \frac{e_ne_2}{s_n^{-1}-s_2} & \dots & \frac{e_n^2}{s_n^{-1}-s_n} \end{pmatrix}\begin{pmatrix}f_1 \\ f_2 \\ \vdots \\ f_n \end{pmatrix}=0. \label{eq:NsolitonGLM}
	\end{align}
%	\begin{align}
%		f_j+e_j-\frac{2\mathrm{i}}{m}\sum_l\frac{e_je_lf_l}{s_j^{-1}-s_l}=0, \qquad (j=1,\dots,n).  \label{eq:proofbs1}
%	\end{align}
	Here, the argument $ x $ is abbreviated.\\
	\indent By using the above notations, the $ n $-soliton solution can be expressed as
	\begin{align}
		\Delta(x) = m+2\mathrm{i}\sum_{j=1}^n s_j^{-1}e_j(x)f_j(x). \label{eq:gapnsol}
	\end{align}
The complex-valued  $ n $-soliton solution has $ 2n $ parameters  $ s_1,\dots,s_n,x_1,\dots,x_n $, and this number of parameters is exactly twice that of the real-valued soliton solution.
	This $ \Delta(x) $ has the following asymptotic form:
	\begin{align}
		\Delta(x) \rightarrow \begin{cases} m & (x\rightarrow-\infty), \\ m\mathrm{e}^{-2\mathrm{i}(\theta_1+\theta_2+\dots+\theta_n)} & (x\rightarrow+\infty).\end{cases}
	\end{align}
	If the solitons are sufficiently separated from each other, the phase shift brought about by the $ j $-th soliton is $ s_j^{-2}=\mathrm{e}^{-2\mathrm{i}\theta_j} $, and the position of the $ j $-th soliton $ X_j $ is given by
	\begin{align}
		X_j=x_j+\frac{1}{\kappa_j}\sum_{\text{$ l $ s.t. $ x_l<x_j $}}\log\left|\frac{\sin\frac{\theta_l+\theta_j}{2}}{\sin\frac{\theta_l-\theta_j}{2}}\right|. \label{eq:SolitonPosition010}
	\end{align}
	Figure \ref{fig:threekink} shows an example of the three-soliton solution.

	The reduction to the real-valued soliton solution is obtained as follows. When the number of solitons is even $ (n=2n') $, the relations
	\begin{align}
		s_{2j-1}=-s_{2j}^*,\quad x_{2j-1}=x_{2j} \quad (j=1,\dots,n') \label{eq:realcond}
	\end{align}
    yield real-valued solutions.
	When the number of solitons is odd $ (n=2n'+1) $, we need to consider the term $ s_{2n'+1}=\mathrm{e}^{\mathrm{i}\pi/2} $ in addition to Eq. (\ref{eq:realcond}) while $ x_{2n'+1} $ remains arbitrary. By this reduction, we obtain $ f_{2j-1}(x)=f_{2j}(x)^* $ and $ f_{2n'+1}(x)=f_{2n'+1}(x)^* $, and the imaginary part of Eq. (\ref{eq:gapnsol}) vanishes. \\
	\indent The  bound state with $ s=s_j \ (\leftrightarrow \ \epsilon=m\cos\theta_j) $ is given by
	\begin{align}
		\begin{pmatrix}u_j(x) \\ v_j(x) \end{pmatrix} = \begin{pmatrix}f_j(x) \\ s_j f_j(x)^*\end{pmatrix} \qquad (j=1,\dots,n). \label{eq:uvnbs}
	\end{align}
	We can show that this state is already normalized, {\it i.e.},  $ \int\mathrm{d}x\bigl(|u_j|^2\!+\!|v_j|^2\bigr)=1 $ holds.\\
	\indent Finally, let $ s $ be real. Then, the scattering states are given by
	\begin{align}
		\begin{pmatrix} u(x,s) \\ v(x,s) \end{pmatrix} = \mathrm{e}^{\mathrm{i}k(s)x}\Biggl[ \begin{pmatrix} 1 \\ s^{-1} \end{pmatrix} +\frac{2\mathrm{i}}{m}\sum_{j=1}^n \frac{e_j(x)}{s_j-s}\begin{pmatrix}f_j(x) \\ s_jf_j(x)^* \end{pmatrix} \Biggr], \label{eq:uvnsol}
	\end{align}
	which are obviously reflectionless as observed from the expression. The amplitudes of these solutions at $ x=\pm\infty $ are
	\begin{align}
		|u(\pm\infty,s)|^2+|v(\pm\infty,s)|^2=1+s^{-2}. \label{eq:ampnsol}
	\end{align}

%\section{Occupation states and gap equation}
	\textit{Occupation states and gap equation.---}%
	\begin{figure}[tb]
		\begin{center}
		\includegraphics{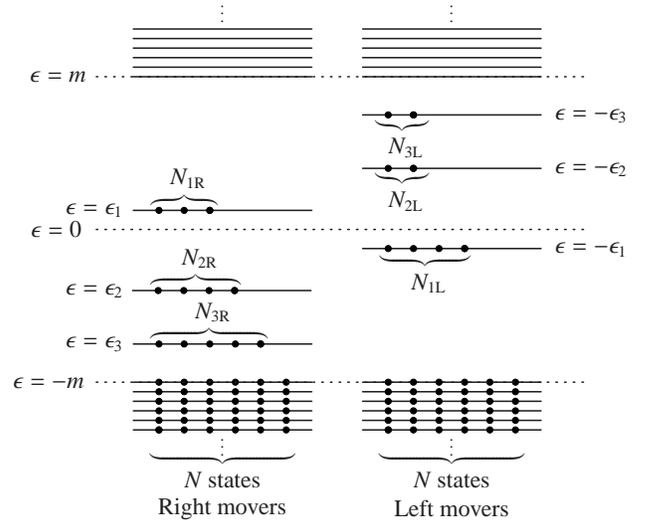}
		\caption{\label{fig:occupation} Diagram of the occupation states considered in this Letter. In this example figure, the number of flavors is $ N=6 $ and the number of solitons is $ n=3 $. The filling rates defined by Eq. (\ref{eq:filling}) are given by $ \nu_1=-1/6,\, \nu_2=1/3 $, and $ \nu_3=1/2 $.}
		\end{center}
	\end{figure}
From this point onwards, we consider the occupation states of the BdG system with $ N $ internal degrees of freedom, or equivalently, the chiral GN model with  $ N $ flavors. We first note that the following relation exists between the solutions of the right and left movers:
	\begin{align}
	\begin{split}
		&(\epsilon, u(x),v(x)) \text{ is a solution of $ \text{BdG}_{\text{R}} $.} \\
		\leftrightarrow \quad& (-\epsilon^*, -v(x)^*,u(x)^*) \text{ is a solution of $ \text{BdG}_{\text{L}} $.}
	\end{split} \label{BdGRLRL}
	\end{align}
	Thus, we can rewrite all quasiparticle wavefunctions of the left movers using those of the right movers. In the light of examining low-energy excited states of condensed matter systems, we consider the configurations in which all the negative-energy scattering states are filled by fermions and positive-energy states are completely vacant. As for bound states, we label the bound states of BdG${}_{\text{R}}$ as  $ (u_{j,\text{R}},v_{j,\text{R}}) \ (j=1,\dots,n) $, and we also label the corresponding bound states of BdG${}_{\text{L}}$ with the energy of the opposite sign related by Eq.~(\ref{BdGRLRL}) as $ (u_{j,\text{L}},v_{j,\text{L}}) $. These states are assumed to be filled partially, and we write the occupation number as  $ N_{j\text{R}} $ and $ N_{j\text{L}} $, as schematically shown in Fig.~\ref{fig:occupation}. The gap equation subsequently becomes
	\begin{align}
	& -\frac{\Delta(x)}{g} \nonumber \\
&= \sum_{\substack{\text{s.s.} \\ \epsilon<0}}Nu_{\text{R}}v_{\text{R}}^*+\sum_{\substack{\text{s.s.} \\  \epsilon<0}}Nu_{\text{L}}v_{\text{L}}^*+\sum_{\text{b.s.}}N_{j\text{R}}u_{j,\text{R}}v_{j,\text{R}}^*+\sum_{\text{b.s.}}N_{j\text{L}}u_{j,\text{L}}v_{j,\text{L}}^* \nonumber \\
		&=N\raisebox{-0.9ex}{\Bigg(}\sum_{\substack{\text{s.s.} \\ \epsilon<0}}u_{\text{R}}v_{\text{R}}^*-\sum_{\substack{\text{s.s.} \\ \epsilon>0}}u_{\text{R}}v_{\text{R}}^* \raisebox{-0.9ex}{\Bigg)}+\sum_{\text{b.s.}}(N_{j\text{R}}-N_{j\text{L}})u_{j,\text{R}}v_{j,\text{R}}^*.
	\end{align}
	Here, the notation s.s. (b.s.) denotes the scattering (bound) states, and the relation (\ref{BdGRLRL}) is used to obtain the second equality in the above equation. Defining the filling rate by
	\begin{align}
		\nu_j := \frac{N_{j\text{R}}-N_{j\text{L}}}{N}, \quad -1\le\nu_j\le1 \qquad (j=1,\dots,n), \label{eq:filling}
	\end{align}
	the above equation can be rewritten as follows:
	\begin{align}
		-\frac{\Delta(x)}{\tilde{g}}=\sum_{\substack{\text{s.s.} \\ \epsilon<0}}u_{\text{R}}v_{\text{R}}^*-\sum_{\substack{\text{s.s.} \\ \epsilon>0}}u_{\text{R}}v_{\text{R}}^*+\sum_{\text{b.s.}}\nu_j u_{j,\text{R}}v_{j,\text{R}}^* \label{eq:gapeqdisc}
	\end{align}
	with  $ \tilde{g}:=Ng $. It is to be noted that the sum of positive-energy scattering states in Eq. (\ref{eq:gapeqdisc}) has a negative sign because of the relation (\ref{BdGRLRL}), and it is equivalent to the stationary condition of the action in the GN model, as given in Ref.~\cite{Shei:1976mn}. Thus, we can again confirm the equivalence of the problems between the BdG and GN systems.\\
	\indent Henceforth, we always use the quantities of the BdG${}_{\text{R}} $ system, and we omit the subscript R. Considering the limit $ L\rightarrow\infty $, where  $ L $ denotes the system size, we replace the sum of the scattering states of the gap equation [Eq.~(\ref{eq:gapeqdisc})] by the corresponding integral. After renormalization of the coupling constant
	\begin{align}
		\frac{1}{\tilde{g}}=\frac{1}{2\pi}\int_{-\infty}^{\infty}\frac{\mathrm{d}k}{\sqrt{k^2+m^2}},
	\end{align}
	and subtracting the logarithmically divergent terms from both sides, we obtain the following expression:
	\begin{align}
		\begin{split}
		&0=\sum_{\text{b.s.}}\nu_j u_{j}(x)v_{j}(x)^*\\
		&\,+\sum_{\epsilon\gtrless 0}\int_{-\infty}^\infty\frac{\mathrm{d}k}{2\pi}\left( \frac{\Delta(x)}{2\sqrt{k^2+m^2}}-(\operatorname{sgn}\epsilon)\frac{u_k(x)v_k(x)^*}{|u_{k,\infty}|^2+|v_{k,\infty}|^2} \right),
		\end{split}
	\end{align}
	where we have written the scattering states with the wavenumber $ k $ as  $ (u_k(x),v_k(x)) $, and their amplitudes at infinity as $ (u_{k,\infty},v_{k,\infty}) $. (We note that  $ (u_k(x),v_k(x)) $ for a positive energy and that for a negative energy are different from each other, though we use the same notation.)
	It is convenient to rewrite the above integral in terms of the uniformizing variable $ s $ introduced in Eq. (\ref{eq:uniformizer}). 
	Using the relation $ \sqrt{k^2+m^2}=\frac{m}{2}|s|(1+s^{-2}) $, we obtain
%	Using the relation (\ref{uniformizer}) and $ \sqrt{k^2+m^2}=\frac{m}{2}|s|(1+s^{-2}) $, we obtain
	\begin{align}
	\begin{split}
		&0=\sum_{\text{b.s.}}\nu_j u_{j}(x)v_{j}(x)^*\\
		&\,+\left[\int_{-\infty}^0\!-\!\int_0^\infty\right]\frac{\mathrm{d}s}{2\pi}\left( \frac{m(1+s^{-2})u(x,s)v(x,s)^*}{2(|u(\infty,s)|^2+|v(\infty,s)|^2)}-\frac{\Delta(x)}{2s} \right). 
	\end{split}\label{eq:gapeqbys}
	\end{align}
	Here, we have written the scattering states labeled by $ s $ as $ (u(x,s),v(x,s)) $.\\ 

%\section[Self-consistent condition for complex  $ n $-soliton solution]{Self-consistent condition for complex  $ \boldsymbol{n} $-soliton solution}
\textit{Self-consistent condition for the $ n $-soliton solution.---}%
	We first present our main result in the following theorem, and then provide the proof.\\
	{\bf Theorem.}
		Let $ \Delta(x) $ be an $ n $-soliton solution given by Eq. (\ref{eq:gapnsol}). The gap equation [Eq.~(\ref{eq:gapeqbys})] holds if and only if the filling rate $ \nu_j $ satisfies 
		\begin{align}
			\nu_j = \frac{2\theta_j-\pi}{\pi} \qquad (j=1,\dots,n).
		\end{align}
%	\end{thm}
	\noindent Here, we remark on certain aspects of this theorem:\\
%	\begin{rem}
1. This theorem provides all self-consistent solutions under the uniform boundary condition [Eq.~(\ref{eq:gapNVB})], because  $ \Delta(x) $  needs to be a reflectionless potential in order for the gap equation to hold \cite{Dashen:1975xh,Shei:1976mn,Feinberg:2003qz}, and $ n $-soliton solutions cover all reflectionless potentials with  $ n $ bound states. \\
2.	The filling rate $ \nu_j $ for the $ j $-th bound state only depends on the phase shift of the $ j $-th soliton, and it is not affected by other soliton parameters. Thus, the self-consistent condition is decoupled for each bound state (or each soliton). \\
%	\end{rem}
%	\begin{rem}
3.	The parameter $ x_j \ (j=1,\dots,n) $, which represents the position of the soliton up to an additive constant [Eq.~(\ref{eq:SolitonPosition010})], is arbitrary and is not related to the self-consistency. This contrasts with the case of real-valued condensates. Because they must be real, the distance between two solitons must be fixed to a specific value, such as in the case of the polarons in polyacetylene \cite{Campbell:1981dc,Campbell:1981,OkunoOnodera} and the topologically trivial soliton in the GN model \cite{Feinberg:2002nq}.\\
%	\end{rem}
%	\begin{rem}
4.	For the $ N $-flavor system, the possible values of the filling rate are given by $ \nu=\frac{N-1}{N},\frac{N-2}{N},\dots,-\frac{N-1}{N} $. Correspondingly, the possible phase shift of each soliton is also discretized. For example,  only the trivial value $ \nu=0 $ is allowed for $ N=1 $, which corresponds to the real kink $ 2\theta=\pi $. When  $ N=2 $, the values  $ \nu=-\frac{1}{2},\,0,\,\frac{1}{2} $ are allowed, which correspond to $ \theta=\frac{\pi}{4},\,\frac{\pi}{2},\frac{3\pi}{4} $. The cases $ N=1 $ and $ 2 $ correspond to $ s $-wave superconductors and polyacetylene, respectively. The cases of $N$ are also obtained as a dimensional reduction of nonrelativistic field theories in 3+1 dimensions, for which $N$ is the number of patches of the Fermi surface \cite{Hofmann:2010gc}. On the other hand, any soliton solution can be self-consistent when $ N=\infty $. 
%	\end{rem}
	\begin{proof}
		Upon substituting Eqs.~(\ref{eq:gapnsol}), (\ref{eq:uvnsol}), and (\ref{eq:ampnsol}) into the integrand of the gap equation (\ref{eq:gapeqbys}),  the terms $ 1+s^{-2} $ and $ |u(\pm\infty,s)|^2+|v(\pm\infty,s)|^2 $ in the first term in the bracket undergo cancellation, thereby yielding 
		\begin{align}
		\begin{split}
			&\frac{m}{2}u(x,s)v(x,s)^*-\frac{\Delta(x)}{2s} \\
			=&\,2\sum_js_j^{-1}e_jf_j\frac{\sin\theta_j}{|s-s_j|^2}-\frac{4\mathrm{i}}{m}\sum_{j,l}\frac{e_jf_je_lf_l}{1-s_js_l}\frac{\sin\theta_j}{|s-s_j|^2} \\ %.\label{eq:proofss1}
			=&\,-2\sum_js_j^{-1}f_j^2\frac{\sin\theta_j}{|s-s_j|^2}.
		\end{split}
		\end{align}
		Here, Eq. (\ref{eq:NsolitonGLM}) is used to obtain the last line.
%From this equation and Eq. (\ref{eq:proofbs1}), 
%we obtain the integrand of the gap equation (\ref{eq:gapeqbys}) as 
%		\begin{align}
%			\frac{m}{2}u(x,s)v(x,s)^*-\frac{\Delta(x)}{2s} = -2\sum_js_j^{-1}f_j^2\frac{\sin\theta_j}{|s-s_j|^2}.
%		\end{align}
%		\begin{gather}
%			\int\frac{\mathrm{d}s}{|s-s_j|^2}=\frac{1}{\sin\theta_j}\tan^{-1}\left( \frac{s-\cos\theta_j}{\sin\theta_j} \right), 
%		\end{gather}
		Using the formula $ \int|s-s_j|^{-2}\mathrm{d}s=(\sin\theta_j)^{-1}\tan^{-1}[(s-\cos\theta_j)/(\sin\theta_j)] $, 
		we can perform the integration and obtain 
		\begin{align}
&			\left[\int_{-\infty}^0-\int_0^\infty\right]\frac{\mathrm{d}s}{2\pi}\left(\frac{m}{2}u(x,s)v(x,s)^*-\frac{\Delta(x)}{2s}\right) \nonumber \\
&= -\sum_js_j^{-1}f_j^2\frac{2\theta_j-\pi}{\pi}.
		\end{align}
		Recalling that the bound states are given by Eq.~(\ref{eq:uvnbs}), we finally obtain
		\begin{align}
			(\text{R.H.S. of Eq. (\ref{eq:gapeqbys})})=\sum_js_j^{-1}f_j^2\left( \nu_j-\frac{2\theta_j-\pi}{\pi} \right).
		\end{align}
		Since the functions  $ f_1(x)^2,\dots,f_n(x)^2 $ are linearly independent of each other, the theorem holds.
	\end{proof}
%\nocite{*}

%\section{Summary}
\textit{Summary.---}%
In summary, we have constructed all the exact self-consistent solutions of complex condensates under uniform boundary conditions.
Our multiple $n$-twisted kink solution 
contains $2n$ parameters, and 
each kink has one bound state. 
Each kink can be placed at any position, while
the self-consistency of the system requires 
the phase shift of each kink to be quantized by $\pi/N$ with the number of flavors $N$.
Our solution describes multiple grey solitons in ultracold atomic fermion gases, and our predictions require experimental verification. 
The dynamics and scattering of these solitons should be studied as a future research topic. 
Further research also needs to be conducted on the construction of self-consistent solutions 
under non-uniform boundary conditions.\\ 
%
%\section{Acknowledgement}
\indent We would like to thank S.~Tsuchiya and R.~Yoshii for a discussion. D.~A.~T. is supported by the JSPS Institutional Program for Young Researcher Overseas Visits. The work of M.~N. is supported in part by 
KAKENHI  (No. 23740198 and No. 23103515). 
%
%\bibliography{complexmultikink-130317}% Produces the bibliography via BibTeX.


\onecolumngrid
\setcounter{equation}{0}
\clearpage
\begin{center}
{\LARGE Supplemental Material}
\vspace{2em}
\end{center}
\twocolumngrid
	In this Supplemental Material, we review the inverse scattering theory of the self-defocusing Zakharov--Shabat (ZS) operator under the finite-density boundary condition \cite{ZakharovShabat2,FaddeevTakhtajan}.  Subsequently, as a special solution, we derive the $ n $-soliton solution and its eigenstates. 
\section{I. Jost Solutions}
	In this section, we consider the direct problem of the ZS operator. We introduce Jost solutions and check several properties of them and transition coefficients defined by their asymptotic form.\\
	\indent The ZS eigenvalue problem in the self-defocusing case is given by
	\begin{align}
		\mathcal{L}\begin{pmatrix}u \\ v\end{pmatrix} =\epsilon\begin{pmatrix}u \\ v\end{pmatrix},\quad \mathcal{L} = \begin{pmatrix} -\mathrm{i}\partial_x & \Delta(x) \\ \Delta(x)^* & \mathrm{i}\partial_x \end{pmatrix}. \label{eq:ZSsupp}
	\end{align}
	We consider this problem under the finite-density boundary condition:
	\begin{align}
		\Delta(x) \rightarrow \begin{cases} m & (x\rightarrow-\infty), \\ m\mathrm{e}^{2\mathrm{i}\theta} & (x\rightarrow+\infty), \end{cases} \qquad (m>0). \label{eq:gapNVBsupp}
	\end{align}
	Since the ZS operator is self-adjoint in the self-defocusing case, 
its eigenvalues $ \epsilon $ are always real. More precisely, if $ (u,v)^T $ is a bounded function, $ \epsilon $ must be real. If we allow exponentially divergent solutions, the solution exists for an arbitrary complex number $ \epsilon $. \\
	\indent When the gap function is constant $ \Delta=m\mathrm{e}^{2\mathrm{i}\theta} $, the dispersion relation of the plane-wave solution $ (u,v)\propto \mathrm{e}^{\mathrm{i}kx} $ becomes 
	\begin{align}
		\epsilon^2 = k^2+m^2. \label{eq:dispe}
	\end{align}
	Thus, the energy spectrum has a gap $ -m<\epsilon<m $. Therefore, if $ \Delta(x) $  satisfies the boundary condition (\ref{eq:gapNVBsupp}), the discrete eigenvalues, which give normalizable bound states, appear in $ |\epsilon|<m $. 
%%%%
\subsection{Wronskian}
	Let $ f_1=(u_1,v_1)^T $ and $ f_2=(u_2,v_2)^T $ be two solutions of Eq. (\ref{eq:ZSsupp}) for a given $ \epsilon $. The Wronskian is defined as 
	\begin{align}
		W(f_1,f_2)=u_1v_2-u_2v_1.
	\end{align}
	It is easy to show that $ \mathrm{d}W/\mathrm{d}x=0 $, i.e.,  $ W $ is a constant. Furthermore, if  $ W\not\equiv 0 $, these two solutions are linearly independent of each other, and therefore, they form a basis of the solution space. 
\subsection{Complex conjugate solution}
	Taking the complex conjugate of Eq. (\ref{eq:ZSsupp}), we obtain the relation
	\begin{align}
		\begin{pmatrix}u \\ v \end{pmatrix}  \text{ is a solution for $ \epsilon $.} \ \Leftrightarrow \  \begin{pmatrix} v^* \\ u^* \end{pmatrix} \text{ is a solution for $ \epsilon^* $.} \label{eq:ccsol}
	\end{align}
\subsection{Uniformizing variable}
	We parametrize the energy $ \epsilon $ and the wavenumber $ k $ using the uniformizing variable $ s $ defined by 
	\begin{align}
		\epsilon(s) = \frac{m}{2}\left( s+s^{-1} \right),\quad k(s) = \frac{m}{2}\left( s-s^{-1} \right). \label{eq:parametrize}
	\end{align}
	The uniformizing variable is convenient because we can avoid to introduce the Riemann surface consisting of two sheets, which is necessary to make the square-root function $ k=\pm\sqrt{\epsilon^2-m^2} $ single-valued. We can confirm that the dispersion relation (\ref{eq:dispe}) holds for arbitrary complex $ s $. The following relations are obvious from the definition:
	\begin{alignat}{2}
		\epsilon(s)&=\epsilon(s^{-1}),&\quad k(s) &= - k(s^{-1}), \label{eq:ekzzi} \\
		\epsilon(s^*)&=\epsilon(s)^*,&\quad k(s^*)&=k(s)^*.
	\end{alignat}
	It is also easy to check 
	\begin{align}
		\operatorname{Im}s\ge0 \Leftrightarrow \operatorname{Im}k(s)\ge0.
	\end{align}
	For real $ \epsilon $, it holds that $ |\epsilon|\ge m \leftrightarrow s \in \mathbb{R} $ and $ |\epsilon|\le m \leftrightarrow |s|=1 $. Therefore, eigenstates corresponding to $ s $ on the real axis are scattering states, and those corresponding to $s$ on the unit circle are bound states. Since $ s $ and $ s^* $ on the unit circle correspond to the same bound state, it is sufficient to consider the unit circle in the upper half-plane when we count the number of bound states. The edges of the continuous spectrum are given by $ \epsilon=\pm m \leftrightarrow s=\pm1 $. See Figure 1 of the main article. 
\subsection{Right Jost solutions}
	The basis of the solutions of Eq. (\ref{eq:ZSsupp}) with $ \epsilon=\epsilon(s) $ under the constant potential $ \Delta(x)=m\mathrm{e}^{2\mathrm{i}\theta} $ can be written as
	\begin{align}
		\begin{pmatrix}s\mathrm{e}^{\mathrm{i}\theta} \\ \mathrm{e}^{-\mathrm{i}\theta}\end{pmatrix}\mathrm{e}^{\mathrm{i}k(s)x},\quad \begin{pmatrix}\mathrm{e}^{\mathrm{i}\theta} \\ s\mathrm{e}^{-\mathrm{i}\theta}\end{pmatrix}\mathrm{e}^{-\mathrm{i}k(s)x}. \label{eq:unifsol}
	\end{align}
	Using the solutions (\ref{eq:unifsol}), we define the right Jost solution $ f_+(x,s) $ as a solution which has the following asymptotic forms:
	\begin{align}
		f_+(x,s)\rightarrow\begin{cases} a(s)\begin{pmatrix}s \\[-0.75ex] 1 \end{pmatrix}\mathrm{e}^{\mathrm{i}k(s)x}+b(s)\begin{pmatrix}1 \\[-0.75ex] s \end{pmatrix}\mathrm{e}^{-\mathrm{i}k(s)x} &(x\rightarrow-\infty) \\ \begin{pmatrix}s\mathrm{e}^{\mathrm{i}\theta} \\[-0.75ex] \mathrm{e}^{-\mathrm{i}\theta} \end{pmatrix}\mathrm{e}^{\mathrm{i}k(s)x} &(x\rightarrow+\infty). \end{cases} \label{eq:rightjost}
	\end{align}
	The solution with the above asymptotic form at $ x\rightarrow +\infty $ is uniquely determined, and hence the transition coefficients $ a(s) $ and $ b(s) $ are also defined uniquely. Since $ \Delta\rightarrow m $ at  $ x\rightarrow-\infty $, the phase factor $ \mathrm{e}^{\mathrm{i}\theta} $ is unnecessary for the form of $ x\rightarrow-\infty $. We do not need to introduce a new symbol for the right Jost solution with wavenumber $ -k $; it can be obtained simply by replacing  $ s\rightarrow 1/s $, because of Eq. (\ref{eq:ekzzi}):
	\begin{align}
		sf_+(x,s^{-1})\!\rightarrow\!\begin{cases} a(s^{-1})\begin{pmatrix}1 \\[-0.75ex] s \end{pmatrix}\mathrm{e}^{-\mathrm{i}k(s)x}\!+\!b(s^{-1})\begin{pmatrix}s \\[-0.75ex] 1 \end{pmatrix}\mathrm{e}^{\mathrm{i}k(s)x} &\!\!\!(x\rightarrow-\infty) \\ \begin{pmatrix}\mathrm{e}^{\mathrm{i}\theta} \\[-0.75ex] s\mathrm{e}^{-\mathrm{i}\theta} \end{pmatrix}\mathrm{e}^{-\mathrm{i}k(s)x} &\!\!\! (x\rightarrow+\infty). \end{cases} \label{eq:rightjost2}
	\end{align}
	Here we summarize the fundamental properties of the right Jost solutions and transition coefficients:
	\begin{enumerate}[(i)]
		\item The set of solutions $ \{ f_+(x,s), sf_+(x,s^{-1}) \} $ are linearly independent of each other unless $ s=\pm1 $, and hence they span the solution space.
		\item The relations $  a(s^{-1*})=a(s)^* $ and $ b(s^{-1*})=b(s)^* $ hold. Specifically,  $ a(s^{-1})=a(s)^* $ and $ b(s^{-1})=b(s)^* $ hold when $ s $ is real. 
		\item The relation $  a(s)a(s^{-1})-b(s)b(s^{-1})=1 $ holds. Specifically, $ |a(s)|^2-|b(s)|^2=1 $ holds when $ s $ is real. 
	\end{enumerate}
	\begin{proof} (i) and the former part of (iii) follow from the evaluation of the Wronskian of these two solutions at $ x=\pm\infty $: 
		\begin{align}
		\begin{split}
			W(+\infty)&=s^2-1 \\
			=W(-\infty)&=(s^2-1)\left[ a(s)a(s^{-1})-b(s)b(s^{-1})\right].
		\end{split}
		\end{align}
		By taking the complex conjugate of Eq.~(\ref{eq:rightjost}) and using the relation (\ref{eq:ccsol}), one can show
		\begin{align}
			\sigma_1 f_+(x,s)^* = s^*f_+(x,s^{-1*}). \label{eq:ccf}
		\end{align}
		Comparing the asymptotic forms of both sides of this equation, one obtains (ii). The latter part of (iii) follows from (ii).
	\end{proof}
	If we write $ t(s)=1/a(s) $ and $ r(s)=b(s)/a(s) $, the latter part of (iii) represents the conservation law $ |t(s)|^2+|r(s)|^2=1 $, and they are interpreted as the transmission and reflection coefficients.
\subsection{Left Jost solutions}
	We define the left Jost solution by the following asymptotic form:
	\begin{align}
		f_-(x,s) \rightarrow \begin{pmatrix} 1 \\ s \end{pmatrix}\mathrm{e}^{-\mathrm{i}k(s)x} \quad (x\rightarrow-\infty).
	\end{align}
	The counterpart with the plus-sign wavenumber $ +k $ can be written as $ s f_-(x,s^{-1}) $ by the same discussion with the right Jost solutions.\\
	\indent Using the asymptotic forms (\ref{eq:rightjost}) and (\ref{eq:rightjost2}), the transformation matrix between the right and left Jost solutions is obtained as
	\begin{align}
		\begin{pmatrix} f_+(x,s) & sf_+(x,s^{-1}) \end{pmatrix} = \begin{pmatrix} sf_-(x,s^{-1}) & f_-(x,s)\end{pmatrix}\begin{pmatrix} a(s) & b(s^{-1}) \\ b(s) & a(s^{-1}) \end{pmatrix}. \label{eq:rjosttoljost}
	\end{align}
	The determinant of the coefficient matrix is unity because of the property (iii). Therefore, the inverse relation is given by
	\begin{align}
		\begin{pmatrix} sf_-(x,s^{-1}) &\! f_-(x,s)\end{pmatrix}\!=\! \begin{pmatrix} f_+(x,s) &\! sf_+(x,s^{-1}) \end{pmatrix}\!\begin{pmatrix} a(s^{-1}) &\! -b(s^{-1}) \\ -b(s) &\! a(s) \end{pmatrix}\!. \label{eq:ljosttorjost}
	\end{align}
	By using Eq. (\ref{eq:ljosttorjost}), we can also write down the asymptotic forms for left Jost solutions.
\subsection{Bound states}
	The discrete spectrum, which gives normalizable bound states, is given by the zeros of the transition coefficient $ a(s) $. As already mentioned, the zeros exist on the unit circle in the $ s $-plane. If the system has $ n $ bound states, there are $ n $ zeros in the upper-half plane, and also $ n $ zeros in the lower-half plane. Since two zeros complex conjugate to each other represent the same bound state, it is sufficient to consider the zeros in the upper-half plane only.\\
	\indent Let the zeros in the upper-half plane and corresponding wavenumbers be $ s_j (j=1,\dots,n) $ and $ k(s_j)=\mathrm{i}\kappa_j $, respectively. Then, Eq. (\ref{eq:rjosttoljost}) becomes
	\begin{align}
		f_+(x,s_j)=b(s_j)f_-(x,s_j), \label{eq:rjostbound}
	\end{align}
	and the asymptotic forms in Eq.~(\ref{eq:rightjost}) are expressed as
	\begin{align}
		f_+(x,s_j)\rightarrow\begin{cases} b(s_j)\begin{pmatrix}1 \\[-0.75ex] s_j \end{pmatrix}\mathrm{e}^{\kappa_j x} &(x\rightarrow-\infty) \\ \begin{pmatrix}s_j\mathrm{e}^{\mathrm{i}\theta} \\[-0.75ex] \mathrm{e}^{-\mathrm{i}\theta} \end{pmatrix}\mathrm{e}^{-\kappa_j x} &(x\rightarrow+\infty). \end{cases} \label{eq:rightjost25}
	\end{align}
	The following Proposition holds for the normalization constant of bound states. 
	\begin{prop}\label{prop:boundstates} Let us write one of bound states as $ f_+(x,s_j)=(u_j(x),v_j(x))^T $ and define the normalization constant $ c_j^2 $ as follows:
	\begin{align}
		c_j^{-2} := \int_{-\infty}^\infty\!\!\mathrm{d}x \left( |u_j(x)|^2+|v_j(x)|^2 \right).
	\end{align}
	By definition $ c_j^2 $ is real, positive and finite. The following relation between the transition coefficients and the normalization constant holds:
	\begin{align}
		\mathrm{i}\dot{a}(s_j)b(s_j)^*s_j=\frac{m}{2}c_j^{-2}. \label{eq:transandnorm}
	\end{align}
	Here, the dot represents the differentiation with respect to $ s $, i.e.,  $ \dot{a}(s_j) = \frac{\mathrm{d} a(s)}{\mathrm{d} s}\Big|_{s=s_j} $. As a corollary of  Eq.~(\ref{eq:transandnorm}), it also follows that $ \dot{a}(s_n)\ne0 $, which implies that all the zeros of $ a(s) $ are simple. 
	\end{prop}
	\begin{proof}
		We often omit arguments of functions when it is clear from the context. The dot and the subscript $ x $ denote the differentiation with respect to $ s $ and  $ x $, respectively. First, we prove the following relation for real $ \epsilon $:
		\begin{align}
			-\mathrm{i}(\dot{u}u^*-\dot{v}v^*)_x=\dot{\epsilon}\left( |u|^2+|v|^2 \right). \label{eq:bslem}
		\end{align}
		It can be shown as follows. Equation (\ref{eq:ZSsupp}) and its derivative with respect to $s$ are 
		\begin{alignat}{2}
			-\mathrm{i}u_x+\Delta v&=\epsilon u,& \quad \mathrm{i}v_x+\Delta^*u&=\epsilon u, \label{eq:bslem2} \\
			-\mathrm{i}\dot{u}_x+\Delta \dot{v}&=\dot{\epsilon}u+\epsilon \dot{u},& \quad \mathrm{i}\dot{v}_x+\Delta^*\dot{u}&=\dot{\epsilon}v+\epsilon\dot{v}, \label{eq:bslem3}
		\end{alignat}
respectively.
		From the complex conjugate of (\ref{eq:bslem2}), we obtain
		\begin{align}
			\Delta = \frac{\epsilon^* v^*+\mathrm{i}v_x^*}{u^*},\quad \Delta^* = \frac{\epsilon^*u^*-\mathrm{i}u_x^*}{v^*}.
		\end{align}
		Using these, we can eliminate  $ \Delta, \Delta^* $ from Eq. (\ref{eq:bslem3}), yielding
		\begin{align}
			-\mathrm{i}(\dot{u}u^*-\dot{v}v^*)_x=\dot{\epsilon}\left( |u|^2+|v|^2 \right)+(\epsilon-\epsilon^*)(\dot{u}u^*+\dot{v}v^*),
		\end{align}
		which reduces to Eq. (\ref{eq:bslem}) if $ \epsilon=\epsilon^* $. On the other hand, 
with recalling that  $ a(s_j)=0 $ and $ k(s_j)=\mathrm{i}\kappa_j $, 
substituting $ s=s_j $ into the differentiation of Eq. (\ref{eq:rightjost}) with respect to $s$ yields 
	\begin{align}
		&\!\!\!\dot{f}_+(x,s_j)\rightarrow \nonumber \\
		&\!\!\!\begin{cases} \dot{a}(s_j)\begin{pmatrix}s_j \\[-0.75ex] 1 \end{pmatrix}\mathrm{e}^{-\kappa_jx}\!+\!\begin{pmatrix}\dot{b}(s_j)-\mathrm{i}b(s_j)\dot{k}(s_j)x \\[-0.75ex] \dot{b}(s_j)s_j+b(s_j)-\mathrm{i}b(s_j)s_j\dot{k}(s_j)x \end{pmatrix}\mathrm{e}^{\kappa_jx} \hspace{1em} \!\!\!\!\!\!\!\!\!\!\!\!\! \\ \hfill (x\rightarrow-\infty)\hphantom{.}\!\!\!\!\!\!\!\!\!\!\!\!\! \\ \begin{pmatrix}s_j\mathrm{e}^{\mathrm{i}\theta}(1+\mathrm{i}s_j\dot{k}(s_j)x) \\[-0.75ex] \mathrm{i}\mathrm{e}^{-\mathrm{i}\theta}\dot{k}(s_j)x \end{pmatrix}\mathrm{e}^{-\kappa_jx} \hfill (x\rightarrow+\infty).\!\!\!\!\!\!\!\!\!\!\!\!\! \end{cases} \label{eq:rightjost3}
	\end{align}
	When $ \epsilon=\epsilon(s_j) $ and $ (u,v)^T=f_+(x,s_j) $, 
one can show from Eqs.~(\ref{eq:rightjost25}) and (\ref{eq:rightjost3}) that
	\begin{align}
		\lim_{x\rightarrow+\infty}\dot{u}u^*-\dot{v}v^*=0,\quad \lim_{x\rightarrow-\infty}\dot{u}u^*-\dot{v}v^*=\dot{a}(s_j)b(s_j)^*(s_j-s_j^*).
	\end{align}
	Therefore, the integration of both sides of Eq. (\ref{eq:bslem}) over $ x $ yields 
		\begin{align}
			\dot{\epsilon}(s_j)c_j^{-2}=\mathrm{i}\dot{a}(s_j)b(s_j)^*(s_j-s_j^*).
		\end{align}
	Since the relations $ \dot{\epsilon}(s)=k(s)/s $ and  $ s_j^*=s_j^{-1} $ hold, we obtain Eq. (\ref{eq:transandnorm}).
	\end{proof}
\subsection{Jost solutions for large  $ \epsilon $}
	When $ |\epsilon| $ is sufficiently large, the contribution of $ \Delta(x) $ becomes relatively negligible, and therefore the Jost solution comes close to a simple plane wave. Actually, the following relations hold for large $ |\epsilon| $:
	\begin{align}
		f_+(x,s)\mathrm{e}^{-\mathrm{i}k(s)x} &= \begin{pmatrix}s\mathrm{e}^{\mathrm{i}\theta} \\ \mathrm{e}^{-\mathrm{i}\theta}\end{pmatrix}+\tilde{O}(|s|), \label{eq:rjostasyms} \\
		f_-(x,s)\mathrm{e}^{\mathrm{i}k(s)x} &= \begin{pmatrix}1 \\ s\end{pmatrix}+\tilde{O}(|s|). \label{eq:ljostasyms}
	\end{align}
	Here  $ \tilde{O}(|s|) $ denotes
	\begin{align}
		\tilde{O}(|s|) = \begin{cases} O(1) & (|s|\gg1) \\ O(|s|) & (|s|\ll1). \end{cases}
	\end{align} 
	These relations are shown by deriving the Volterra integral equation from Eq. (\ref{eq:ZSsupp}) and solving it iteratively. 
	The transition coefficient  $ a(s) $ can be written as
	\begin{align}
		a(s)= \frac{W(f_+(x,s),f_-(x,s))}{s^2-1} \label{eq:asinwronskian}
	\end{align}
	by using Eqs.~(\ref{eq:rightjost}) and (\ref{eq:ljosttorjost}). From Eqs. (\ref{eq:rjostasyms}), (\ref{eq:ljostasyms}) and (\ref{eq:asinwronskian}), we obtain
	\begin{align}
		a(s)=\begin{cases} \mathrm{e}^{\mathrm{i}\theta}+O(|s|^{-1}) & (|s|\gg 1) \\ \mathrm{e}^{-\mathrm{i}\theta}+O(|s|) & (|s|\ll1). \end{cases}\label{eq:asasyms}
	\end{align}
\subsection{ $ a(s) $ expressed in terms of scattering data}
	\indent Let us define the following function in the region  $ \operatorname{Im}s\ge0 $:
	\begin{align}
		\tilde{a}(s) = \mathrm{e}^{-\mathrm{i}\theta}a(s)\prod_{j=1}^n\frac{s-s_j^*}{s-s_j}.
	\end{align}
	By definition, $ \tilde{a}(s) $ has no zero and no pole in the upper-half plane, and  $ \tilde{a}(s)=1+O(|s|^{-1}) $ for $ |s|\gg1 $. Therefore the function $ \log \tilde{a}(s) $ is analytic and satisfies $ \log\tilde{a}(s)=O(|s|^{-1}) $ for large $ |s| $ in the upper-half plane. From the Cauchy's integral formula, the relation
	\begin{align}
		\log \tilde{a}(s) = \frac{1}{\pi\mathrm{i}}\int_{-\infty}^\infty\!\!\mathrm{d}z\frac{\log|\tilde{a}(z)|}{z-s} \quad (\operatorname{Im}s>0)
	\end{align}
	follows. Rewriting this expression in terms of $ a(s) $ by using the relation $ |\tilde{a}(z)|^2=|a(z)|^2=|t(z)|^{-2}=(1-|r(z)|^2)^{-1} $, we obtain
	\begin{align}
		a(s) = \mathrm{e}^{\mathrm{i}\theta}\prod_{j=1}^n\frac{s-s_j}{s-s_j^*} \exp\left[ \frac{1}{2\pi \mathrm{i}}\int_{-\infty}^\infty\!\!\mathrm{d}z\frac{\log(1-|r(z)|^2)}{s-z} \right]
	\end{align}
	for $ \operatorname{Im}s>0 $. This relation shows how to determine $ a(s) $ from the scattering data. 
	In the limit $ s\rightarrow0 $, we obtain
	\begin{align}
		\mathrm{e}^{2\mathrm{i}\theta}=\prod_{j=1}^n\frac{s_j^*}{s_j} \exp\left[ \frac{1}{2\pi \mathrm{i}}\int_{-\infty}^\infty\!\!\mathrm{d}z\frac{\log(1-|r(z)|^2)}{z} \right]
	\end{align}
	with the use of Eq. (\ref{eq:asasyms}). As a special case, if the reflection coefficient vanishes identically, it reduces to
	\begin{align}
		\mathrm{e}^{2\mathrm{i}\theta}=\prod_{j=1}^n\frac{s_j^*}{s_j},
	\end{align}
	which provides the derivation of Eq. (15) of the main article.
%%%%
\section{II. Inverse problem}
	In this section, we introduce the integral representation of Jost solutions using the integral kernel $ K(x,y) $. Subsequently, we derive the Gel'fand--Levitan--Marchenko (GLM) equation, that is, the integral equation which determines $ K(x,y) $ from scattering data.
\subsection{Integral representation of Jost solutions}
	Let us assume that the left Jost solution can be expressed in terms of a  $ 2\times 2 $ matrix integral kernel $ K(x,y) $ as follows:
	\begin{gather}
		f_-(x,s)=\begin{pmatrix} 1 \\ s \end{pmatrix}\mathrm{e}^{-\mathrm{i}k(s)x}+\int_{-\infty}^x\!\!\mathrm{d}y K(x,y)\begin{pmatrix} 1 \\ s \end{pmatrix}\mathrm{e}^{-\mathrm{i}k(s)y}, \label{eq:lJostkernelK}
	\end{gather}
where the integrand should not diverge at $ y=-\infty $. If $ K(x,y) $ decreases exponentially in the limit $ y\rightarrow-\infty $, the expression (\ref{eq:lJostkernelK}) is well defined for $ \operatorname{Im}k\ge0 \ \leftrightarrow \ \operatorname{Im}s\ge0 $. By replacement $ s\rightarrow1/s $, we obtain the similar expression for the other left Jost solution
	\begin{align}
		sf_-(x,s^{-1})=\begin{pmatrix} s \\ 1 \end{pmatrix}\mathrm{e}^{\mathrm{i}k(s)x}+\int_{-\infty}^x\!\!\mathrm{d}y K(x,y)\begin{pmatrix} s \\ 1 \end{pmatrix}\mathrm{e}^{\mathrm{i}k(s)y}. \label{eq:lJostkernelK2}
	\end{align}
	This expression, on the other hand, is well defined for $ \operatorname{Im}s\le0  $. We note that both expressions in Eqs.~(\ref{eq:lJostkernelK}) and (\ref{eq:lJostkernelK2}) are simultaneously well defined only when $ s $ is real.\\
	\indent By the same logic of deriving Eq. (\ref{eq:ccf}),
	\begin{align}
		\sigma_1f_-(x,s)^*=s^*f_-(x,s^{-1*})
	\end{align}
	follows, and we obtain the following relation from this relation and 
Eq.~(\ref{eq:lJostkernelK}):
	\begin{align}
		K(x,y)=\sigma_1 K(x,y)^* \sigma_1. \label{eq:Kinvolt}
	\end{align}
	In each component, we obtain $ K_{22}=K_{11}^* $ and $  K_{12}=K_{21}^* $. Thus, only the first column of $ K(x,y) $ is independent, and it has the following form:
	\begin{align}
			K(x,y)=\begin{pmatrix} K_{11}(x,y) & K_{21}(x,y)^* \\ K_{21}(x,y) & K_{11}(x,y)^*  \end{pmatrix}.
	\end{align}

%%%%%%%%%%
\subsection{Equations for $ K(x,y) $}
	\begin{prop} The integral kernel $ K(x,y) $ satisfies the following equations:
	\begin{gather}
		K(x,x)-\sigma_3K(x,x)\sigma_3= U(x)-M, \label{eq:KtoU} \\
		\frac{\partial K(x,y)}{\partial x}+\sigma_3\left( \frac{\partial K(x,y)}{\partial y}-K(x,y)M \right)\sigma_3-U(x)K(x,y)=0. \label{eq:KtoU2}
	\end{gather}
	Here, we have introduced the following notations:
	\begin{align}
		U(x)&:= \begin{pmatrix} 0 & -\mathrm{i}\Delta(x) \\ \mathrm{i}\Delta(x)^* & 0 \end{pmatrix}, \\
		M&:= U(-\infty)=m\sigma_2.
	\end{align}
	In particular, Eq.~(\ref{eq:KtoU}) can be used to construct the potential $ \Delta(x) $:
	\begin{align}
		\Delta(x)=m+2\mathrm{i}K_{21}(x,x)^*. \label{eq:KtoU3}
	\end{align}
	\end{prop}
	\begin{proof}In order to make descriptions short, we introduce the following temporary notations:
		\begin{align}
			F(x,s)&=\begin{pmatrix} sf_-(x,s^{-1}) & f_-(x,s) \end{pmatrix}, \\
			Z(x,s)&=\begin{pmatrix} s\mathrm{e}^{\mathrm{i}k(s)x} & \mathrm{e}^{-\mathrm{i}k(s)x} \\ \mathrm{e}^{\mathrm{i}k(s)x} & s\mathrm{e}^{-\mathrm{i}k(s)x} \end{pmatrix}.
		\end{align}
		$ F(x,s) $ and $ Z(x,s) $ satisfy the following ZS eigenvalue problem:
		\begin{align}
			\partial_xF(x,s)&=\Bigl(\mathrm{i}\epsilon(s)\sigma_3+U(x)\Bigr)F(x,s), \label{eq:ljostinproof} \\
			\partial_xZ(x,s)&=\Bigl(\mathrm{i}\epsilon(s)\sigma_3+M\Bigr)Z(x,s). \label{eq:fndinproof}
		\end{align}
		Equations (\ref{eq:lJostkernelK}) and (\ref{eq:lJostkernelK2}) can be expressed as
		\begin{align}
			F(x,s) = Z(x,s)+\int_{-\infty}^x\!\!\mathrm{d}yK(x,y)Z(y,s). \label{eq:kerninproof}
		\end{align}
		Henceforth, we omit the argument $ s $ and simply write them as $ \epsilon=\epsilon(s) $, $ F(x)=F(x,s) $, and $ Z(x)=Z(x,s) $. By differentiating Eq. (\ref{eq:kerninproof}) with respect to $ x $ and using Eq.~(\ref{eq:fndinproof}), we obtain
		\begin{align}
			\partial_xF(x)=(\mathrm{i}\epsilon\sigma_3+M)Z(x)+K(x,x)Z(x)+\!\!\int_{-\infty}^x\!\!\mathrm{d}y\frac{\partial K(x,y)}{\partial x}Z(y). \label{eq:ljostinproof2}
		\end{align}
		Next, let us rewrite the R.H.S. of (\ref{eq:ljostinproof}). It follows that
		\begin{align}
		\begin{split}
			& \mathrm{i}\epsilon\sigma_3\!\!\int_{-\infty}^x\!\!\mathrm{d}yK(x,y)Z(y) \\
			=& \int_{-\infty}^x\!\!\mathrm{d}y\sigma_3K(x,y)\sigma_3(\mathrm{i}\epsilon\sigma_3Z(y)) \\
			=& \int_{-\infty}^x\!\!\mathrm{d}y\sigma_3K(x,y)\sigma_3(\partial_yZ(y)-MZ(y)) \\
			=& \sigma_3 K(x,x)\sigma_3 Z(x)+\int_{-\infty}^x\!\!\mathrm{d}y\sigma_3\left( K(x,y)M-\frac{\partial K(x,y)}{\partial y} \right)\sigma_3Z(y).
		\end{split} \label{eq:ljostinproof6}
		\end{align}
		Here, we have used $ \sigma_3^2=1 $ in the second line, Eq. (\ref{eq:fndinproof}) in the third line, and  $ \sigma_3 M = -M\sigma_3 $ in the last line. Using Eqs. (\ref{eq:kerninproof}) and (\ref{eq:ljostinproof6}), the R.H.S. of (\ref{eq:ljostinproof}) can be rewritten as
		\begin{align}
			&(\mathrm{i}\epsilon\sigma_3+U(x))F(x)\nonumber \\
			=\ &(\mathrm{i}\epsilon\sigma_3+U(x)+\sigma_3K(x,x)\sigma_3)Z(x) \nonumber\\
			&\!+\!\int_{-\infty}^x\!\!\mathrm{d}y\left[\sigma_3\left( K(x,y)M-\frac{\partial K(x,y)}{\partial y} \right)\sigma_3+U(x)K(x,y)\right]Z(y). \label{eq:ljostinproof3}
		\end{align}
		From $ \text{(R.H.S of (\ref{eq:ljostinproof2}))}= \text{(R.H.S of (\ref{eq:ljostinproof3}))} $, we obtain
		\begin{align}
		\begin{split}
			&\sigma_3K(x,x)\sigma_3-K(x,x)+U(x)-M \\
			=&\int_{-\infty}^x\!\!\mathrm{d}y\Biggl[ \frac{\partial K(x,y)}{\partial x}+\sigma_3\left( \frac{\partial K(x,y)}{\partial y}-K(x,y)M \right)\sigma_3\\
			&\qquad\qquad\qquad\qquad -U(x)K(x,y) \Biggr] Z(y,s)Z(x,s)^{-1}. 
		\end{split}\label{eq:ljostinproof4}
		\end{align}
		Here, we again write the $s$-dependence of $ Z(x,s) $ explicitly. 
In Eq.~(\ref{eq:ljostinproof4}), the L.H.S. is a function dependent only on $ x $. On the other hand, the R.H.S. depends on  $ x $ and $ s $. Therefore, the L.H.S. vanishes when we differentiate both sides with respect to $ s $. In order for the relation $ \partial_s\text{(R.H.S.)}=0 $ to hold for any $ x $ and $ s $, the integrand must vanish identically, we thus obtain Eq.~(\ref{eq:KtoU2}). As a result, the L.H.S. of (\ref{eq:ljostinproof4}) also vanishes, and Eq.~(\ref{eq:KtoU}) follows.
	\end{proof}
%%%%%%%
\subsection{GLM equation}
	\begin{thm} Let  $ s_j \ (j=1,\dots,n) $ be zeros of $ a(s) $ in the upper-half plane. The equation
	\begin{align}
		K(x,y)\begin{pmatrix}1 \\ 0\end{pmatrix}+F(x+y)+\int_{-\infty}^x\!\!\mathrm{d}zK(x,z)F(z+y)=0 \label{eq:GLM}
	\end{align}
	holds for $ y<x $. Here $ F(x) $ is defined by
	\begin{align}
		F(x) &:= F_c(x)+F_d(x), \\
		F_c(x) &:= \frac{m}{4\pi}\int_{-\infty}^\infty\!\!\mathrm{d}s\,r(s)\begin{pmatrix}s^{-1} \\ 1\end{pmatrix}\mathrm{e}^{-\mathrm{i}k(s)x}, \label{eq:dfFc}\\
		F_d(x) &:=\sum_{j=1}^n |b(s_j)|^2 c_j^2\begin{pmatrix} 1 \\ s_j \end{pmatrix}\mathrm{e}^{\kappa_j x}, \label{eq:dfFd}
	\end{align}
	where  $ r(s)=b(s)/a(s) $ is a reflection coefficient,  $ \kappa_j=-\mathrm{i}k(s_j) $ is a complex wavenumber, and $ c_j^2 $ is the normalization constant introduced in Proposition \ref{prop:boundstates}.
	\end{thm}
	\indent Before giving the proof, we remark that this theorem does not contain the case of $ y=x $, since the discussion on the convergence becomes rather sensitive in this case. Indeed, if one tries to include the case $ y=x $ in the theorem, one finds that the proof is not valid in a concrete example of the $ n $-soliton solution discussed in the next section. However, we must use the case of $y=x$ in order to construct the potential $ \Delta(x) $, as shown in Eq.~(\ref{eq:KtoU3}). The safest way to overcome this dilemma  is to interpret $ K(x,x) $ as a limiting value, i.e.,  $ K(x,x):=\lim_{y\rightarrow x-0}K(x,y) $. 
	\begin{proof}
	\indent From Eq. (\ref{eq:rjosttoljost}), the equation
	\begin{align}
		\frac{1}{a(s)}f_+(x,s)= sf_-(x,s^{-1})+\frac{b(s)}{a(s)}f_-(x,s)
	\end{align}
holds. 
	By rewriting this by using Eqs.~(\ref{eq:lJostkernelK}) and (\ref{eq:lJostkernelK2}), we obtain
	\begin{align}
	\begin{split}
		&\frac{1}{a(s)}f_+(x,s)-\begin{pmatrix}s \\ 1 \end{pmatrix}\mathrm{e}^{\mathrm{i}k(s)x}\\
		&=\int_{-\infty}^x\!\!\mathrm{d}z K(x,z)\begin{pmatrix}s \\ 1\end{pmatrix}\mathrm{e}^{\mathrm{i}k(s)z}\\
		&\qquad+\frac{b(s)}{a(s)}\left[\begin{pmatrix}1 \\ s\end{pmatrix}\mathrm{e}^{-\mathrm{i}k(s)x}+\int_{-\infty}^x\!\!\mathrm{d}zK(x,z)\begin{pmatrix}1 \\ s\end{pmatrix}\mathrm{e}^{-\mathrm{i}k(s)z}\right].
	\end{split}\label{eq:beforeint}
	\end{align}
	Henceforth, we calculate
	\begin{align}
		\frac{m}{2}\int_{-\infty}^\infty\!\!\mathrm{d}s\frac{\mathrm{e}^{-\mathrm{i}k(s)y}}{s}(\text{Eq. (\ref{eq:beforeint}})) \qquad (y< x). \label{eq:beforeint2}
	\end{align}
	\indent First, let us confirm that this integral has a finite value. From Eq. (\ref{eq:rjostasyms}), the integrand of L.H.S. of Eq. (\ref{eq:beforeint2}) for large and small $ |s| $ can be estimated as follows:
	\begin{align}
	\begin{split}
		&\left[\frac{1}{a(s)}f_+(x,s)-\begin{pmatrix}s \\ 1 \end{pmatrix}\mathrm{e}^{\mathrm{i}k(s)x}\right]\frac{\mathrm{e}^{-\mathrm{i}k(s)y}}{s} \\
		&=\begin{cases} O(|s|^{-1})\times \mathrm{e}^{\mathrm{i}\frac{m}{2}s(x-y)} & (|s|\gg 1)\hphantom{.} \\ O(1)\times\mathrm{e}^{-\mathrm{i}\frac{m}{2s}(x-y)} &(|s|\ll 1). \end{cases}
	\end{split}
	\end{align}
	Therefore, the integration for large  $ |s| $ converges, unless $ x-y=0 $. Though we encounter the rapidly oscillating function around $ s=0 $, the integration gives a finite result.\\
	\indent Let us calculate the L.H.S. of Eq. (\ref{eq:beforeint2}). Since $ \mathrm{e}^{\mathrm{i}k(s)(x-y)} $ decreases exponentially in the upper-half plane, we can use the residue theorem. Here we note that all zeros of $ a(s) $ are simple, as mentioned in Proposition \ref{prop:boundstates}. As usual, considering the contour consisting of real axis and semicircle with radius $ R $ in the upper-half plane, and taking the limit $ R\rightarrow\infty $, we obtain
	\begin{align}
		(\text{L.H.S. of Eq. (\ref{eq:beforeint2})}) &= 2\pi \mathrm{i}\sum_{j=1}^n \frac{m}{2}\frac{1}{\dot{a}(s_j)}s_j^{-1}f_+(x,s_j)\mathrm{e}^{-\mathrm{i}k(s_j)y} \nonumber \\
		&=-2\pi\sum_{j=1}^n |b(s_j)|^2c_j^2 f_-(x,s_j)\mathrm{e}^{\kappa_jy}.
	\end{align}
	Here, we have used Eqs. (\ref{eq:rjostbound}) and (\ref{eq:transandnorm}) to show the second equality. Furthermore, using the expression (\ref{eq:lJostkernelK}) and the definition (\ref{eq:dfFd}), we obtain
	\begin{align}
		\frac{(\text{L.H.S. of Eq. (\ref{eq:beforeint2})})}{2\pi}=-F_d(x+y)-\int_{-\infty}^x\!\!\mathrm{d}zK(x,z)F_d(z+y). \label{eq:LHSeFd}
	\end{align}
	\indent Next, let us consider the R.H.S. of Eq. (\ref{eq:beforeint2}). Let us note the formulae
	\begin{align}
		\frac{m}{2}\int_{-\infty}^\infty\!\!\mathrm{d}s\mathrm{e}^{\mathrm{i}k(s)x}&=2\pi \delta(x), \label{eq:intexpk} \\
		\frac{m}{2}\int_{-\infty}^\infty\!\!\mathrm{d}s\frac{\mathrm{e}^{\mathrm{i}k(s)x}}{s}&=0 \label{eq:intexpk2}
	\end{align}
which can be shown by dividing the integral into two regions,  $ [0,\infty] $ and  $ [-\infty,0] $, and substituting $ \tilde{s}=-s^{-1} $ in the latter integral. 
By using these formulae, the expression 
	\begin{align}
	\begin{split}
		\frac{(\text{R.H.S. of Eq. (\ref{eq:beforeint2})})}{2\pi} &= K(x,y)\begin{pmatrix}1 \\ 0\end{pmatrix}+F_c(x+y)\\
		&\quad+\int_{-\infty}^x\!\!\mathrm{d}zK(x,z)F_c(z+y)
	\end{split}\label{eq:RHSeFc}
	\end{align}
	follows by straightforward calculation. From Eqs. (\ref{eq:LHSeFd}) and (\ref{eq:RHSeFc}), we obtain the theorem. 
	\end{proof}

%%%%%%%%%%
\section{III. Multi-Soliton Solution and Eigenstates}
	In this section we solve the GLM equation (\ref{eq:GLM}) when the reflection coefficient vanishes identically: $r(s)=0$. All the expressions for the  $ n $-soliton solution and its eigenstates shown in the main article without proof are provided here.
%%%%
\subsection{$ n $-soliton solution}
	\indent We consider the case where the potential has $ n $ bound states. With letting $ C_j=|b(s_j)|c_j(>0) $, the form of $ F(x) $ is given by
	\begin{align}
		F(x) = \sum_{j=1}^n C_j^2\begin{pmatrix}1 \\ s_j \end{pmatrix}\mathrm{e}^{\kappa_j x}. \label{eq:reflectionlessF}
	\end{align}
	Let $ K_1(x,y) $ be the first column of  $ K(x,y) $. Then the second column can be written as $ \sigma_1 K_1(x,y)^* $ from Eq. (\ref{eq:Kinvolt}). When $ F(x) $ has the form (\ref{eq:reflectionlessF}), the solution of the GLM equation can be obtained by the following ansatz:
	\begin{align}
		K_1(x,y)=\sum_{j=1}^n C_j\begin{pmatrix} f_j(x) \\ s_jf_j(x)^* \end{pmatrix}\mathrm{e}^{\kappa_jy}. \label{eq:ansatzK1}
	\end{align}
	Substituting Eqs. (\ref{eq:reflectionlessF}) and (\ref{eq:ansatzK1}) into Eq. (\ref{eq:GLM}) yields
	\begin{align}
		&\sum_{j=1}^nC_j\mathrm{e}^{\kappa_jy}\begin{pmatrix}F_j(x) \\ s_jF_j(x)^* \end{pmatrix}=0, \label{eq:reflectionlessGLM} \\
		&F_j(x) := f_j(x)+C_j\mathrm{e}^{\kappa_j x}+\sum_{l=1}^n \frac{(1+s_l^{-1}s_j)C_lC_j\mathrm{e}^{(\kappa_l+\kappa_j)x}}{\kappa_l+\kappa_j}f_l(x).
	\end{align}
	By taking the complex conjugate in the second component of Eq.~(\ref{eq:reflectionlessGLM}), we obtain
	\begin{align}
		\sum_{j=1}^nC_j\begin{pmatrix}1 \\ s_j^* \end{pmatrix}\mathrm{e}^{\kappa_jy}F_j(x)=0. \label{eq:reflectionlessGLM2}
	\end{align}
	Since $ (\mathrm{e}^{\kappa_jy}, s_j^*\mathrm{e}^{\kappa_jy})^T \ (j=1,\dots,n) $ are linearly independent of each other, all $ F_j(x) $'s must vanish.\\
	\indent The numerical factor appearing in $ F_j(x) $ can be rewritten as 
	\begin{align}
		\frac{1+s_l^{-1}s_j}{\kappa_l+\kappa_j}=-\frac{2\mathrm{i}}{m}\frac{1}{s_j^{-1}-s_l}.
	\end{align}
	In order to simplify descriptions, we also introduce the following symbol:
	\begin{align}
		 e_j(x)=C_j \mathrm{e}^{\kappa_jx}, \quad (j=1,\dots,n).
	\end{align}
	If we parametrize $ C_j=\sqrt{\kappa_j}\mathrm{e}^{-\kappa_jx_j} $, it is equivalent to Eq. (12) of the main article. The equation $ F_j(x)=0 \ (j=1,\dots,n) $ can be represented as follows:
	\begin{align}
		\begin{pmatrix}f_1 \\ f_2 \\ \vdots \\ f_n \end{pmatrix}+\begin{pmatrix}e_1 \\ e_2 \\ \vdots \\ e_n \end{pmatrix}-\frac{2\mathrm{i}}{m}\begin{pmatrix} \frac{e_1^2}{s_1^{-1}-s_1} & \frac{e_1e_2}{s_1^{-1}-s_2} & \dots & \frac{e_1e_n}{s_1^{-1}-s_n} \\ \frac{e_2e_1}{s_2^{-1}-s_1} & \frac{e_2^2}{s_2^{-1}-s_2} & \dots & \frac{e_2e_n}{s_2^{-1}-s_n} \\ \vdots && \ddots & \vdots \\ \frac{e_ne_1}{s_n^{-1}-s_1} & \frac{e_ne_2}{s_n^{-1}-s_2} & \dots & \frac{e_n^2}{s_n^{-1}-s_n} \end{pmatrix}\begin{pmatrix}f_1 \\ f_2 \\ \vdots \\ f_n \end{pmatrix}=0. \label{eq:NsolitonGLMsupp}
	\end{align}
	Solving this equation, $ K_1(x,y) $ can be obtained as
	\begin{align}
		K_1(x,y)=\sum_{j=1}^n \begin{pmatrix} f_j(x) \\ s_jf_j(x)^* \end{pmatrix}e_j(y),
	\end{align}
	and using this, the potential $ \Delta(x) $ is constructed as
	\begin{align}
		\Delta(x) = m+2\mathrm{i}K_{21}(x,x)^*=m+2\mathrm{i}\sum_{j=1}^n s_j^{-1}e_j(x)f_j(x).
	\end{align}
%%%%%
\subsection{Eigenstates}
	Since $ K_1(x,y) $ is given by Eq.~(\ref{eq:ansatzK1}), the left Jost solution (\ref{eq:lJostkernelK}) can be evaluated as follows:
	\begin{align}
		f_-(x,s)&=\mathrm{e}^{-\mathrm{i}k(s)x}\left[ \begin{pmatrix} 1 \\ s \end{pmatrix} +\sum_{j=1}^n C_j(1+s_j^{-1}s)\begin{pmatrix}f_j(x) \\ s_j f_j(x)^* \end{pmatrix}\frac{\mathrm{e}^{\kappa_j x}}{\kappa_j-\mathrm{i}k(s)} \right] \nonumber \\
		&=\mathrm{e}^{-\mathrm{i}k(s)x}\left[ \begin{pmatrix} 1 \\ s \end{pmatrix} +\frac{2\mathrm{i}}{m}\sum_{j=1}^n \frac{e_j(x)}{s_j-s^{-1}}\begin{pmatrix}f_j(x) \\ s_j f_j(x)^* \end{pmatrix} \right]. \label{eq:ljostforNsol}
	\end{align}
	By replacement $ s\rightarrow1/s $ in Eq. (\ref{eq:ljostforNsol}), we also obtain 
	\begin{align}
		f_-(x,s^{-1})=\mathrm{e}^{\mathrm{i}k(s)x}\left[ \begin{pmatrix} 1 \\ s^{-1} \end{pmatrix} +\frac{2\mathrm{i}}{m}\sum_{j=1}^n \frac{e_j(x)}{s_j-s}\begin{pmatrix}f_j(x) \\ s_j f_j(x)^* \end{pmatrix} \right]. \label{eq:nsolss}
	\end{align}
	It represents the scattering states if $ s $ is real. \\
	\indent The bound states can be obtained by substitution  $ s=s_l \ (l=1,\dots,n) $ in Eq. (\ref{eq:ljostforNsol}). Because of Proposition \ref{prop:boundstates}, the normalization constant of  $ f_-(x,s_l) $ is equal to  $ C_l=|b(s_l)|c_l $. We thus obtain the \textit{normalized} bound states 
	\begin{align}
		-C_lf_-(x,s_l)&=-e_l(x)\left[ \begin{pmatrix} 1 \\ s_l \end{pmatrix} +\frac{2\mathrm{i}}{m}\sum_{j=1}^n \frac{e_j(x)}{s_j-s_l^{-1}}\begin{pmatrix}f_j(x) \\ s_j f_j(x)^* \end{pmatrix} \right] \nonumber \\
		&=\begin{pmatrix} f_l(x) \\ s_lf_l(x)^* \end{pmatrix}. \label{eq:nsolbs}
	\end{align}
	Here, we have used Eq.~(\ref{eq:NsolitonGLMsupp}) to show the second equality. Equations (\ref{eq:nsolss}) and (\ref{eq:nsolbs}) provide the eigenstates shown in the main article without derivation.

%%%%%%%%%
\subsection{Positions of solitons}
	Finally, let us derive the approximate expression for positions of solitons when they are well-separated from each other. As in the main article, we use the convention $e_j(x)=\sqrt{\kappa_j}\mathrm{e}^{\kappa_j(x-x_j)}$. For convenience, we prepare the notations for 1-soliton solution with phase shift $ s_1^{-2}=\mathrm{e}^{-2\mathrm{i}\theta_1} $ located at $ x=0 $ and its bound state:
	\begin{align}
	\begin{split}
		\Delta_{\text{1-sol}}(x,s_1)&= m\frac{\mathrm{e}^{\kappa_1x-2\mathrm{i}\theta_1}+\mathrm{e}^{-\kappa_1x}}{\mathrm{e}^{\kappa_1x}+\mathrm{e}^{-\kappa_1x}} \\
		&=m\mathrm{e}^{-\mathrm{i}\theta_1}\left(\cos\theta_1-\mathrm{i}\sin\theta_1\tanh\kappa_1x  \right),
	\end{split} \\
		f_{\text{1-sol}}(x,s_1)&=-\frac{\sqrt{\kappa_1}}{2}\frac{1}{\cosh\kappa_1x}.
	\end{align}
	Our purpose is to show the following Proposition.
	\begin{prop}
	Let us assume that all $ x_j $'s are sufficiently separated from each other. In this situation, we can relabel the indices of solitons so that $ x_1\ll x_2\ll \dotsb \ll x_n $ holds. The approximate position of the $ j $-th soliton $ X_j $ is given as follows:
	\begin{align}
		X_1 &=x_1, \\
		\begin{split}
		X_j &= x_j+\frac{1}{\kappa_j}\sum_{l=1}^{j-1}\log\left|\frac{1-s_ls_j}{s_l-s_j}\right| \\
			&= x_j+\frac{1}{\kappa_j}\sum_{l=1}^{j-1}\log\left|\frac{\sin\frac{\theta_l+\theta_j}{2}}{\sin\frac{\theta_l-\theta_j}{2}}\right| \quad (j\ge2).
		\end{split}
	\end{align}
	The approximate expressions of the potential $ \Delta(x) $ and the $ j $-th bound state near $ x\simeq X_j $ are given by
	\begin{align}
		\Delta(x) &\simeq \mathrm{e}^{-2\mathrm{i}(\theta_1+\dotsb+\theta_{j-1})}\Delta_{\textrm{{\rm 1}-{\rm sol}}}(x-X_j,s_j), \label{eq:solposigap1}\\
		f_j(x) & \simeq \mathrm{e}^{-\mathrm{i}(\theta_1+\dotsb+\theta_{j-1})}\operatorname{sgn}\Biggl(\prod_{l=1}^{j-1}(\theta_l-\theta_j)\Biggr)f_{\textrm{{\rm 1}-{\rm sol}}}(x-X_j,s_j). \label{eq:solposibs1}
	\end{align}
	\end{prop}
	\begin{proof}
		Since the bound states are localized to each of the solitons, all $ f_l(x) $'s except for $ f_j(x) $ are negligible near the $ j $-th soliton. Furthermore, $ e_l(x) \ (l\ge j+1) $ are also negligible since $ x_j-x_l\ll0 $. Thus, Eq.~(\ref{eq:NsolitonGLMsupp}) approximately reduces to the following equation near the $ j $-th soliton :
	\begin{align}
		\begin{pmatrix}0 \\ \vdots \\ 0 \\ f_j \end{pmatrix}+\begin{pmatrix}e_1 \\ e_2 \\ \vdots \\ e_j \end{pmatrix}-\frac{2\mathrm{i}}{m}\begin{pmatrix} \frac{e_1^2}{s_1^{-1}-s_1} & \frac{e_1e_2}{s_1^{-1}-s_2} & \dots & \frac{e_1e_j}{s_1^{-1}-s_j} \\ \frac{e_2e_1}{s_2^{-1}-s_1} & \frac{e_2^2}{s_2^{-1}-s_2} & \dots & \frac{e_2e_j}{s_2^{-1}-s_j} \\ \vdots && \ddots & \vdots \\ \frac{e_je_1}{s_j^{-1}-s_1} & \frac{e_je_2}{s_j^{-1}-s_2} & \dots & \frac{e_j^2}{s_j^{-1}-s_j} \end{pmatrix}\begin{pmatrix}f_1 \\ f_2 \\ \vdots \\ f_j \end{pmatrix}=0. \label{eq:NsolitonGLM2}
	\end{align}
	Here, we note that even though $ f_l(x) $ for $ l<j $  is exponentially small, the product $ f_l(x) e_l(x) $ must be kept in the above equation because $ e_l(x) $ increases exponentially and $ f_l(x) e_l(x) \sim O(1) $ in the limit $ x\rightarrow+\infty $. The solution of Eq. (\ref{eq:NsolitonGLM2}) is given by
	\begin{align}
		\!\!f_l(x)\!= \frac{1}{e_l(x)}\frac{\prod_{q=1}^{j-1}(s_q^{-1}\!-\!s_l)}{\prod_{p=1,p\ne l}^j(s_p\!-\!s_l)}\left( \frac{m(s_j^{-1}\!-\!s_l)}{2\mathrm{i}}+\frac{\kappa_j^2}{\kappa_j+t_j(x)^2}\! \right) \label{eq:solposi01}
	\end{align}
	for $ l=1,\dots,j $, where we have defined $ t_j(x) $ as
	\begin{align}
		t_j(x):=&\,e_j(x)\prod_{l=1}^{j-1}\frac{s_l-s_j}{1-s_ls_j}\nonumber \\
		=&\sqrt{\kappa_j}\mathrm{e}^{\kappa_j(x-X_j)}\operatorname{sgn}\Biggl(\prod_{l=1}^{j-1}(\theta_l-\theta_j)\Biggr). \label{eq:solposit}
	\end{align}
	We can confirm that Eq. (\ref{eq:solposi01}) solves Eq. (\ref{eq:NsolitonGLM2}) by direct substitution and using the following identity:
	\begin{align}
		\sum_{l=1}^j\frac{\prod_{q=1}^k(y_q-x_l)}{\prod_{p=1,p\ne l}^j(x_p-x_l)}=\begin{cases} 1 & (k=j-1) \\ 0 & (0\le k< j-1). \end{cases} \label{eq:corlagrange}
	\end{align}
	Here, $ x_1,\dots, x_j $ are complex numbers different from each other and $ y_1,\dots,y_k $ are arbitrary complex numbers. We note that Eq. (\ref{eq:corlagrange}) is a corollary of the Lagrange interpolation formula 
	\begin{align}
		\prod_{q=1}^k(x-y_q)=\!\sum_{l=1}^j\prod_{q=1}^k(x_l-y_q)\!\prod_{p=1,p\ne l}^j\frac{x-x_p}{x_l-x_p} \quad (0\le k\le j-1).
	\end{align}
	Comparing the coefficient of $ x^{j-1} $, we obtain Eq.~(\ref{eq:corlagrange}).\\ 
	\indent When $ l=j $, Eq. (\ref{eq:solposi01}) is simplified as
	\begin{align}
		f_j(x)= -\Biggl(\prod_{p=1}^{j-1}s_p^{-1}\Biggr)\frac{\kappa_jt_j(x)}{\kappa_j+t_j(x)^2}. \label{eq:solposibs2}
	\end{align}
Using Eqs. (\ref{eq:solposi01}) and (\ref{eq:corlagrange}), $ \Delta(x) $ can be obtained as
	\begin{align}
		\Delta(x) = m+2\mathrm{i}\sum_{l=1}^js_l^{-1}e_lf_l = m\,\Biggl(\prod_{p=1}^{j-1}s_p^{-2}\Biggr)\frac{\kappa_j+s_j^{-2}t_j^2}{\kappa_j+t_j^2}. \label{eq:solposigap2}
	\end{align}
	Equations (\ref{eq:solposibs2}) and (\ref{eq:solposigap2}) with (\ref{eq:solposit}) are equivalent to Eqs. (\ref{eq:solposigap1}) and (\ref{eq:solposibs1}), and these expressions imply that  $ X_j $ represents  the position of the $ j $-th soliton.
	\end{proof}

\end{document}